\def\D{\cal D}\def\I{\cal I}

\def\={=\!\!}

\documentclass{article}
\usepackage{amsthm} 
\usepackage{graphicx} 
\usepackage{amssymb,latexsym,times}
\newcommand{\psfrag}[2]{}
\usepackage{makeidx,proof}
\usepackage{enumerate}
\newtheorem{theorem}{Theorem}
\newtheorem{definition}[theorem]{Definition}
\newtheorem{lemma}[theorem]{Lemma}
\newtheorem{proposition}[theorem]{Proposition}
\newtheorem{corollary}[theorem]{Corollary}

\title{Dependence and Independence\thanks{Research partially supported by the EUROCORES LogICCC LINT programme.}
}
\author{Erich Gr\"adel\\ 
Mathematische Grundlagen der Informatik,\\ RWTH Aachen University,
 52056 Aachen, Germany \and Jouko V\"a\"an\"anen\thanks{Research partially supported by
grant 40734 of the Academy of Finland.}\\ Department of Mathematics and Statistics\\ University of Helsinki, Finland\\ and \\
Institute for Logic, Language and Computation \\ University of Amsterdam, The Netherlands}

\begin{document}
\maketitle
\def\vx{\vec{x}}
\def\vy{\vec{y}}
\def\vz{\vec{z}}
\def\vu{\vec{u}}
\def\vv{\vec{v}}
\def\vw{\vec{w}}
\def\boto{\ \bot\ }
\newcommand{\idep}[3]{#1\ \bot_{#2}\ #3}
\newcommand{\idepb}[2]{#1\ \bot\ #2}
\def\yxz{\idep{\vy}{\vx}{\vz}}
\def\xzy{\idep{\vx}{\vz}{\vy}}
\def\xzx{\idep{\vx}{\vz}{\vx}}
\def\xyx{\idep{\vx}{\vy}{\vx}}
\def\xzu{\idep{\vx}{\vz}{\vu}}
\def\uzy{\idep{\vu}{\vz}{\vy}}
\def\yzy{\idep{\vy}{\vz}{\vy}}
\def\yxy{\idep{\vy}{\vx}{\vy}}
\def\xyu{\idep{\vx}{\vy}{\vu}}
\def\xy{\idep{\vx}{\vy}}
\def\dyx{\=(\vy,\vx)}

\begin{abstract}
We introduce an atomic formula $\yxz$  intuitively saying that the variables $\vy$ are independent from the variables $\vz$ if the variables $\vx$ are kept constant. We contrast this with  dependence logic $\D$ \cite{MR2351449} based on the atomic formula $\=(\vx,\vy)$, actually equivalent to  $\yxy$, saying that the variables $\vy$ are totally determined by the variables $\vx$. We show that $\yxz$ gives rise to a natural logic capable of formalizing basic intuitions about independence and dependence.  We show that $\yxz$ can be used to give partially ordered quantifiers and IF-logic an alternative interpretation without some of the shortcomings related to so called signaling that interpretations using $\=(\vx,\vy)$ have. 

\end{abstract}

Of the numerous uses of the word ``dependence" we focus on the concept of an attribute\footnote{color, price, salary, height, etc} depending on a number of other similar attributes when we observe the world. We call these attributes variables.
We follow the approach of \cite{MR2351449} and focus on the strongest form of dependence, namely functional dependence. This is  the kind of dependence in which some given variables absolutely deterministically determine some variables, as surely as $x$ and $y$ determine $x+y$ and $x\cdot y$ in elementary arithmetic. The idea is that weaker forms of dependence can be understood in terms of the strongest. Functional dependence of $x$ on $\vy$ is denoted in \cite{MR2351449} by the symbol
 $\=(\vy,{x})$.  If we adopt the shorthand
$$\=(\vy,\vx)\mbox{ for }\=(\vy,{x_1})\wedge\ldots\wedge\=(\vy,{x_n})$$
we get a more general functional dependence. Although there are many different intuitive meanings for 
$\dyx$, such as ``$\vy$ totally determines $\vx$" or ``$\vx$ is a function of $\vy$", the best way to understand the concept is to give it semantics: 
\begin{definition}[\cite{MR98k:03068,MR2351449}]\label{def}
 Sets of assigments are called {\em teams}. A team  $X$ satisfies $\=(\vy,\vx)$ in $M$, in symbols 
 $M\models_X \ \=(\vy,\vx) $, or just $X\models\ \=(\vy,\vx)$, if \begin{equation}\label{funct}\forall s,s'\in X(s(\vy)=s'(\vy)\to s(\vx)=s'(\vx)).\end{equation}
\end{definition}
Condition (\ref{funct}) is a universal statement. As a consequence it is closed downward, that is, if a team satisfies it, every subteam does. In particular, the empty team satisfies it for trivial reasons. Also, every singleton team $\{s\}$ satisfies
it, again for trivial reasons.

Functional dependence has been studied in database theory and some basic properties, called {\bf Armstrong's Axioms} have been isolated \cite{armstrong}. These axioms state the following properties of $\=(\vy,\vx)$:

\begin{enumerate}
\item $\=(\vx,\vx)$. Anything is functionally dependent of itself.
\item If $\=(\vy,\vx)$ and $\vy\subseteq\vz$, then $\=(\vz,\vx)$. Functional dependence is preserved by increasing input data.
\item If $\vy$ is a permutation of $\vz$, $\vu$ is a permutation of $\vx$, and $\=(\vz,\vx)$, then $\=(\vy,\vu)$. Functional dependence does not look at the order of the variables.
\item If $\=(\vy,\vz)$  and $\=(\vz,\vx)$, then  $\=(\vy,\vx)$. Functional dependences can be transitively composed. 
\end{enumerate} 

These rules  completely describe the behavior of $\=(\vy,\vx)$  in the following sense: If $T$ is a finite set of dependence atoms of the form $\=(\vy,\vx)$ for various $\vx$ and $\vy$, then $\=(\vy,\vx)$ follows from $T$ according to the above rules if and only if every team that satisfies $T$ also satisfies $\=(\vy,\vx)$. Let us see how Armstrong \cite{armstrong} proved this: Suppose $T\models\ \=(\vy,\vx)$, i.e. every team satisfying $T$ satisfies $\=(\vy,\vx)$. Let $\vz\supseteq\vy$ be the list of variables $z$ such that $\=(\vy,z)$ can be derived using the rules (1)-(4) from $T$. Let $x\in \vx$. We show $x\in \vz$. Suppose not. Let $X$ be the team $\{s,s'\}$, where $s(z)=s'(z)=0$ for all $z\in \vz$, but $s(u)=0, s'(u)=1$ for all $u\not\in\vz$. Note that $X\not\models\ \=(\vy,\vx)$ because $x\not\in\vz$.   So it suffices to show $X\models T$. Suppose $\=(\vu,\vv)\in T$. Let $v\in \vv$. We show $X\models\ \=(\vu,v)$. If $\vu\cap -\vz\ne\emptyset$, then $s(\vu)\ne s'(\vu)$. So w.l.o.g. $\vu\subseteq\vz$. If $v$ is not $x$, then $s(v)=s'(v)$. So w.l.o.g. $v$ is the variable $x$. But now transitivity gives $x\in \vz$, contradiction. QED

We shall now give the concept of independence a similar treatment as we gave above to the concept of dependence. Again we start from the strongest conceivable form of independence of variables $x$ and $y$, a kind of total lack of connection between them, which we denote $x\bot y$.
We can read this in many ways:

\begin{itemize}
\item  $x$ and $y$ are completely independent from each other.
\item $x$ and $y$ occur totally freely.
\item $x$ and $y$ give absolutely no information of each other.
\item Every conceivable pattern occurs for $x$ and $y$.
\end{itemize}

Suppose balls of different sizes and masses are dropped from the Leaning Tower of Pisa in order to observe how the size and mass influence the time of descent. One may want to make sure that in this test:
\begin{equation}\label{wea3}\begin{array}{l}
\mbox{\it } \mbox{\it The size of the ball is independent of the mass of the ball.} 
\end{array}\end{equation} 
How to make sure of this? Ideally one would vary the sizes and the masses freely so that if one mass is chosen for one size it would be also be chosen for all the other sizes, and if one size is chosen for one mass it is also chosen for all other masses. This would eliminate any dependence between size and mass and the test would genuinely tell us something about the time of descent itself.  We would then say that the size and the mass were made independent of each other in the strongest sense of the word.

Suppose we have data about tossing two coins and we want to state:
\begin{equation}\label{coin}\begin{array}{l}
\mbox{\it Whether one coin comes heads up is independent}\\
\mbox{\it  of whether the other coin comes heads up.}
\end{array}\end{equation} 
To be convinced, one should look at the data and point out that all four possibilities occur. Probability theory has its own concept of independence which however is in harmony with ours, only we do not pay attention to how many times certain pattern occurs.  In probability theory, roughly speaking,
two random variables are independent if observing one does not affect the (conditional) probability of the other. We could say the same without paying attention to probabilities as follows: two variables are independent if observing one does not restrict in any way what the value of the other is.
 
If we look at any demographic data except for rather small data we may observe:
\begin{equation}\label{wea4}\begin{array}{l}
\mbox{\it A person's gender is independent} \\
\mbox{\it of whether the person speaks Spanish. }\\
\end{array}\end{equation} 
We would use a given data as support of  the truth of this by finding in the  data a male and a female that speak Spanish, and a male and a female that do not speak Spanish. Once this is established it would be rather difficult to claim that there is some dependence in the given data between the gender and the ability to speak Spanish. Of course this analysis again ignores the probabilities. That is, the data may have many females who speak Spanish but only one male that speaks Spanish and still our requirement for independence would be satisfied. This just shows that our criterion is really of a {\em logical} kind, not of a probabilistic kind.

When Galileo dropped balls of the same size from the Leaning Tower of Pisa he was able to observe:

\begin{equation}
\begin{array}{l}
\mbox{\it Their time of descent is independent of their mass.}
\end{array}\end{equation} 
What did this mean? It means that each ball has the same time of descent, a constant, and therefore Galileo could conclude that it is independent of the mass. Being constant is a  kind of strong independence different from the above examples where we emphasized that in a sense all possible patterns should occur. When Galileo dropped balls from the tower he did not observe all possible patterns and still he was able to conclude a certain independence.

Einstein stated in his theory of special relativity that:
\begin{equation}
\begin{array}{l}
\mbox{\it  The speed of light is independent of the observer's state of motion.}
\end{array}\end{equation} 
This is another famous example of independence where one of the variables is constant. Of course the constancy of the speed of light was  not considered a scientific fact at the time, although observations supported it. 

%
%
%
So we should accept that one form of total independence is when one of the variables is a constant.

Another feature of the strongest possible independence is symmetry. In our example (2)-(6) there is a total symmetry of the variables. There are weaker forms of independence where symmetry is not present.  For example: The result of collecting data about trading might support the claim:
\begin{equation}\label{wea1}\begin{array}{l}
\mbox{\it This investor's trading is independent of non-public}\\
\mbox{\it information about the company.}
\end{array}\end{equation} 
However, there would be no reason to believe that as a consequence:
\begin{equation}\label{wea2}\begin{array}{l}
\mbox{\it Non-public information about the company is independent of }\\
\mbox{\it this investor's trading.}
\end{array}\end{equation} 
%

%

Let us now introduce the semantics of $x\boto y$:

\begin{definition}\label{depdef}
A team $X$ satisfies the atomic formula $x\boto y$ if \begin{equation}
\label{indep}\forall s,s'\in X\exists s''\in X(s''(y)=s(y)\wedge s''(x)=s'(x)).
\end{equation} 
\end{definition}

What this definition says is the following criterion for a team $X$ of ``data" to manifest the independence of $x$ and $y$: Knowing $s(x)$ alone for a given $s\in X$ gives no information about $s(y)$, because there may be $s'\in X$ with $s'(y)\ne s(y)$, and then (\ref{indep}) gives $s''\in X$ with $s''(x)=s(x)$ and $s''(y)=s'(y)$. So just when we were going to say that $s(x)$ is enough evidence to conclude that the value of $y$ is $s(y)$, we see this other $s''$ with the same value for $x$ but a different value for $y$.

We can immediately observe that a constant variable is independent of every other variable, including itself. To see this, suppose $x$ is constant in $X$. Let $y$ be any variable, possibly $y=x$. If $s,s'\in X$ are given, we need $s''\in X$ such that $s''(x)=s(x)$ and $s''(y)=s'(y)$. We can simply take $s''=s'$. Now $s''(x)=s(x)$, because $x$ is constant in $X$. Of course, $s''(y)=s'(y)$. Conversely, if $x$ is independent of every other variable, it is clearly constant, for it would have to be independent of itself, too. So we have $$\=(x)\iff x\boto x.$$

We can also immediately observe the symmetry of independence, because the criterion (\ref{indep}) is symmetrical in $x$ and $y$. More exactly, $s''(y)=s(y)\wedge s''(x)=s'(x)$ and  $s''(x)=s'(x)\wedge s''(y)=s(y)$ are trivially equivalent.

Our observations on constancy and symmetry lead to the following definition:

\begin{definition} The following rules are called the {\bf Independence Axioms}
\begin{enumerate}
\item If $x\boto y$, then $y\boto x$ (Symmetry Rule).
\item If $x\boto x$, then $y\boto x$ (Constancy Rule).
\end{enumerate} 
\end{definition}

It may seem that independence must have much more content than what these two axioms express, but they are actually complete in the following sense: 

\begin{theorem}[Completeness of the Independence Axioms]
If $T$ is a finite set of dependence atoms of the form $u\boto v$ for various $u$ and $v$, then $y\boto x$ follows from $T$ according to the above rules if and only if every team that satisfies $T$ also satisfies $y\boto x$. 

\end{theorem}

\begin{proof}

Let us see how this follows: Suppose $T\models\ y\boto x$. We try to derive $y\boto x$ from $T$. If 
$y\boto x\in T$ or $x\boto y\in T$, we are done by the Symmetry Rule. So we assume $y\boto x\notin T$ and $x\boto y\notin T$. Let $V$ be the set of variables $z$ such that $z\boto z\in T$. If $x\in V$ or $y\in V$, we are done by the Constancy Rule, so we assume $V\cap\{x,y\}=\emptyset$.   Consider a domain consisting of $V$ and two new elements $0$ and $1$.    
For $d\in\{0,1\}$ let $X_d$ consist of all $s$ such that if $v\in V$, then $s(v)=v$, and moreover $s(y)=s(x)=d.$
Finally, let $X=X_0\cup X_1$. 

Let us first observe that $X\not\models x\boto y$, because there are $s\in X$ with $s(x)=0$ and $s'\in X$ with $s'(y)=1$ but there is no $s''\in X$ such that both $s''(x)=0$ and $s''(y)=1$. 

Let us then show that $X\models T$. Suppose $u\boto v\in T$. If $u\in V$ or $v\in V$, then clearly $X\models u\boto v$. So we may assume $\{u,v\}\cap V=\emptyset$. In particular, $u\ne v$. Suppose $u=x$, $u=y$, $v=x$ or $v=y$. By symmetry we may assume $u=x$. We may then also assume $v\notin \{x,y\}$ for otherwise we are done. Let now $s,s'\in X$ be arbitrary. 
Let $s''(u)=s(u)$, $s''(v)=s'(v)$ and $s''(w)=w$ for $w\in V$. Then $s''\in X$. So we have proved 
$X\models u\boto v$ in this case. The final case is that $\{u,v\}\cap\{x,y\}=\emptyset$. In this case it is trivial that $X\models u\boto v$. 

\end{proof}

The independence atom $y\boto x$ turns out to be a special case of the more general notion
$$\yxz$$
the intuitive meaning of which is that the variable $\vy$ are totally independent of the variables $\vz$ when the variables $\vx$ are kept fixed.

Suppose objects of different forms (balls, pins, etc), different sizes and different masses are dropped from the Leaning Tower of Pisa in order to observe how the form, size and mass influence the time of descent. One may want to make sure that in this test:
\begin{equation}\label{gal}\begin{array}{l}
\mbox{\it } \mbox{\it For a fixed form, the size of the object is}\\ 
\mbox{\it independent of the mass of the object.} 
\end{array}\end{equation} 
How to make sure of this? Ideally one would vary for each form separately the sizes and the masses freely so that if one mass is chosen in that form for one size it would be also be chosen in that form for all the other sizes, and so on.  We would then say that the size and the mass were made independent of each other, given the form, in the strongest sense of the word.

We now give exact mathematical  content to $\yxz$: 
\begin{definition}\label{free}
A team $X$ satisfies the atomic formula $\vy\ \bot_{\vx}\ \vz$ if for all $s,s'\in X$ such that $s(\vx)=s'(\vx)$ there exists $s''\in X$ such that $s''(\vx)=s(\vx)$, $s''(\vy)=s(\vy)$, and $s''(\vz)=s'(\vz))$.

\end{definition}

In the case of the sentence (\ref{gal}) this means the following: A set of observation concerning the falling objects is said to satisfy the requirement (\ref{gal}) if for any two tests $s$ and $s'$ where the form of the objects was the same there is a test $s''$ still with the same form but which picks the size from test $s$ and the weight from test $s'$. Note that this is in harmony with there having been just one test, but of course no scientific experiment would be satisfactory with just one test. So when there are several tests the requirement of (\ref{gal}) being satisfied actually pushes the number of tests up. 

Here are some elementary properties of $\yxz$:

\begin{lemma}
 $\=(\vx,\vy)$ logically implies $\vy\ \bot_{\vx}\ \vz$. 
\end{lemma}

\begin{proof}
Suppose $X$ satisfies $\=(\vx,\vy)$. To prove that $X$ satisfies 
$\vy\ \bot_{\vx}\ \vz$, let $s,s'\in X$ such that $s(\vx)=s'(\vx)$. Note that $s(\vy)=s'(\vy)$. We can choose $s''=s'$, for then
$s''\in X$,  $s''(\vx)=s(\vx)$, $s''(\vy)=s(\vy)$, and $s''(\vz)=s'(\vz)$.
\end{proof}

\begin{lemma}
 $\vy\ \bot_{\vx}\ \vz$ logically implies $\=(\vx,\vy\cap\vz)$.
\end{lemma}

\begin{proof}
Suppose 
$X$ satisfies $\vy\ \bot_{\vx}\ \vz$. To prove that $X$ satisfies $\=(\vx,\vy\cap\vz)$, let $s,s'\in X$ such that  
 $s(\vx)=s'(\vx)$. We show $s(\vy\cap\vz)=s'(\vy\cap\vz)$. Let us choose $s''\in X$ such that $s''(\vx)=s(\vx)$, $s''(\vy)=s(\vy)$, and $s''(\vz)=s'(\vz)$. Let $w\in \vy\cap\vz$. 
Then $s(w)=s''(w)=s'(w)$.
\end{proof}

\begin{corollary}
$\=(\vx,\vy)\iff \vy\ \bot_{\vx}\ \vy$
\end{corollary}

So dependence is just a special case of independence, when independence is defined in the more general form. This has the pleasant consequence that when we define {\em independence logic} $\I$ by adding the atomic formulas $\yxz$ to first order logic, we automatically include all of dependence logic. 

We get the following reformulation of the corollary:

\begin{corollary}
$\vy\ \bot_{\vx}\ \vy\Rightarrow\vy\ \bot_{\vx}\ \vz$ (Constancy Rule)
\end{corollary}



Here are some rather trivial properties

\begin{lemma} 
\begin{enumerate}
\item $\vx\ \bot_{\vx}\ \vy$ (Reflexivity Rule)
\item   $\vz\ \bot_{\vx}\ \vy\Rightarrow\vy\ \bot_{\vx}\ \vz$\  (Symmetry Rule)
\item $\vec{y}{y'}\ \bot_{\vec{x}}\ \vec{z}{z'}\Rightarrow\vec{y}\ \bot_{\vec{x}}\ \vec{z}$. (Weakening  Rule)
\item If $\vec{z'}$ is a permutation of $\vz$, $\vec{x'}$ is a permutation of $\vx$, $\vec{y'}$ is a permutation of $\vy$, then $\vy\ \bot_{\vx}\ \vz\Rightarrow\vec{y'}\ \bot_{\vec{x'}}\ \vec{z'}$. (Permutation Rule) 
\end{enumerate} 
\end{lemma}
 
A little less trivial are the following properties:

\begin{lemma} 
\begin{enumerate}
\item   $\vz\ \bot_{\vx}\ \vy\Rightarrow\vy\vx\ \bot_{\vx}\ \vz\vx$\  (Fixed Parameter Rule)
\item $\xzy\wedge\vec{u}\ \bot_{\vec{z}\vec{x}}\ \vec{y}\Rightarrow\uzy$. (First Transitivity  Rule)
\item $\yzy\wedge\vec{z}\vec{x}\ \bot_{\vec{y}}\ \vec{u}\Rightarrow\xzu$ (Second Transitivity  Rule\footnote{We are indebted to P. Constantinou and P. Dawid for pointing out an error in our original formulation.})
\end{enumerate} 
\end{lemma}

Note that the Second Transitivity Rule gives by letting $\vu=\vx$:
$$\yzy\wedge\xyx\Rightarrow\xzx,$$
which is the transitivity axiom of functional dependence. In fact Armstrong's Axioms are all derivable from the above rules. It remains open whether our rules permit a completeness theorem like Armstrong's Axioms do\footnote{We are grateful to F. Engstr\"om for pointing out that this problem is already solved in a slightly different formulation in database theory \cite{CH}.}.

Now we can define a new logic by adding the independence atoms $\yxz$ to first order logic just as dependence logic $\D$ was defined in \cite{MR2351449}:

\begin{definition}
We define {\em independence logic} $\I$ as the extension of first order logic by the new atomic formulas $$\yxz$$ for all sequences $\vy,\vx,\vz$ of variables. The negation sign $\neg$ is allowed in front of atomic formulas. The other logical operations are $\wedge,\vee,\exists$ and $\forall$. The semantics is defined for the new atomic formulas as in Definition~\ref{free} and in other cases exactly as for dependence logic in \cite{MR2351449}. The negation of $\yxz$ is satsified by the empty team alone. 
\end{definition}

There is an obvious alternative game-theoretic semantics based on the idea that a winning strategy should allow ``mixing" of plays in the same way Definition~\ref{free} mixes assignments $s$ and $s'$ into a new one $s''$. As we see below, this means that the existential player cannot use her own moves to code signals to herself and thereby go around requirements of imperfect information.

\def\rel{\mathop{\rm rel}}

Let us recall the following characterization of dependence logic\footnote{A team  $X$ can be thought of as a relation $\rel(X)$ by identifying an assignment $s$ with domain $\{x_1,...,x_n\}$ with the $n$-tuple $(s(x_1),...,s(x_n))$. With this in mind, we use the notation $M\models\Phi(X)$ for $(M,\rel(X))\models\Phi(S)$.}:

\begin{theorem}[\cite{kont}]\label{kont}The expressive power of formulas $\phi(x_1,...,x_n)$ of dependence logic is exactly that of existential second order sentences with the predicate for the team negative. More exactly, let us fix a vocabulary $L$ and an $n$-ary predicate symbol $S\notin L$. Then:
\begin{itemize}
\item For every $L$-formula $\phi(x_1,...,x_n)$ of dependence logic  there is an existential second order $L\cup\{S\}$-sentence $\Phi(S)$, with $S$ negative only, such that for all $L$-structures $M$  and all teams $X$:
\begin{equation}\label{eq}
M\models_X\phi(x_1,...,x_n)\iff M\models\Phi(X).
\end{equation}
\item For every existential second order $L\cup\{S\}$-sentence $\Phi(S)$, with $S$ negative only,  there exists an $L$-formula $\phi(x_1,...,x_n)$ of dependence logic  such that (\ref{eq}) holds  for all $L$-structures $M$ and all teams $X\ne\emptyset$. 
\end{itemize} \end{theorem}

The question arises, whether there is a similar characterization for independence logic. We do not know the answer. However, we may note the following:

\begin{proposition}The expressive power of formulas $\phi(x_1,...,x_n)$ of independence logic is contained in that of existential second order sentences with a predicate $S$ for the team. More exactly, let us fix a vocabulary $L$ and an $n$-ary predicate symbol $S\notin L$. Then for every $L$-formula $\phi(x_1,...,x_n)$ of independence logic  there is an existential second order $L\cup\{S\}$-sentence $\tau_{\phi}(S)$ such that for all $L$-structures $M$  and all teams $X$: $M\models_X\phi(x_1,...,x_n)\iff M\models\tau_{\phi}(X)$.
\end{proposition}

\begin{proof}
The construction of $\tau_{\phi(x_1,...,x_n)}(S)$ is done by induction on $\phi(x_1,...,x_n)$. The construction is done exactly as in  \cite[Theorem 6.2]{MR2351449} except for 
%
the independence atoms. Consider an independence
atom
\[  x_i \bot_{x_j} x_k. \]
In this case $\tau_{\phi(x_1,\dots,x_n)}(S)$ would be:
\begin{center}
\begin{tabular}{l}
$\forall y_1\dots y_n \forall z_1\dots z_n
((S(\vec{y})\land S(\vec{z})\land y_j=z_j)\to$\\
\hspace{15mm}$\exists u_1\dots \exists u_n (S(\vec{u})\land u_j=y_j\land u_i=y_i \land u_k=z_k))$.\\ 
\end{tabular}
\end{center}
It is obvious how to generalize this to independence atoms
among tuples $(x_i)_{i\in I}$,  $(x_j)_{j\in J}$ and
$(x_k)_{k\in K}$ for any index sets $I,J,K\subseteq \{1,\dots,n\}$.
\end{proof}

For sentences $\phi$  of $\I$ we define as for $\D$: $M\models\phi\iff M\models_{\{\emptyset\}}\phi$.
 
\begin{corollary}
For sentences independence logic and dependence logic are equivalent in expressive power.
\end{corollary}

\begin{proof} 
Suppose $\phi$ is a sentence of independence logic. There is  an existential second order sentence $\tau_\phi(S)$ such that for every model $M$  we have $$M\models\phi\iff M\models\tau_\phi(\{\emptyset\}).$$ By Theorem~\ref{kont},  there is a sentence $\psi$ of dependence logic such that for every model $M$  we have $$M\models\tau_\phi(\{\emptyset\})\iff M\models\psi.$$ Thus the sentences $\phi$ and $\psi$ are equivalent.
\end{proof}

Note that formulas of independence logic need not be closed downward (i.e. $M\models_X\phi$ and $Y\subseteq X$ need not imply $M\models_Y\phi$), for example $x\boto y$ is not. This is a big difference to dependence logic. Still, the empty team satisfies every independence formula.

 The sentence $$\forall x\forall y\exists z(z\boto x\wedge z=y)$$ is valid in harmony with the intuition that the existential player should be able to make a decision to be independent of $x$ when she chooses $z$ whether she lets $z=y$ or not.

The sentence $$\forall x\exists y\exists z(z\boto x\wedge z=x)$$ is not valid in harmony with the intuition that the existential player needs to follow what the universal player is doing with his $x$ in order to be able to hit $z=x$. In independence friendly logic (\cite{MR97j:03005}) the sentence 
$$\forall x\exists y\exists z/x(z=x),$$ is valid which is often found counter-intuitive. The trick (called ``signaling") is that the existential player stores the value of $x$ into $y$ and then chooses $z$ on the basis of $y$, apparently not needing to know what $x$ is. In fact one might consider the entire independence friendly logic with the following interpretation:
\begin{equation}
\label{faith}
[\exists x/\vy\phi(x,\vy,\vz)]^*=\exists x(\vy\ \bot_{\vz}\ x\wedge [\phi(x,\vy,\vz)]^*)\end{equation} as an alternative to the usual one based on $\=(\vz,\vx)$.
As we have seen above the interpretation (\ref{faith}) is not necessarily entirely faithful. However,  the atom 
$\vy\ \bot_{\vz}\ x$ has one clearly distinguishable meaning of independence of $\vy$ from $x$ so it might be interesting to look at independence friendly logic with this interpretation. 

Our independence atom works well also in giving partially ordered quantifiers 
compositional semantics as the following lemma illustrates: 

\begin{lemma}Suppose $\phi(x,y,u,v,\vec{z})$ is a first order formula. 
Then the following conditions are equivalent:
\begin{description}
\item[(1)]
$M\models_s\left(\begin{array}{cc}
\forall x&\exists y\\
\forall u&\exists v
\end{array}\right)\phi(x,y,u,v,\vec{z})$
\item[(2)]$M\models_{\{s\}}\forall x\exists y\forall u\exists 
v(v\ \bot_{u\vz}\ x\wedge\phi(x,y,u,v,\vec{z}))$
\end{description}
\end{lemma}

The implication from (1) to (2) is obvious. The converse is 
a consequence of the fact that the implication
\[ (\=(x,u,v) \land v\bot_u x)\ \Rightarrow\ \=(u,v)\]
is valid in all teams.
Suppose not. Then there is a team $X$ satisfing the left side,
with two assigments $s,s'\in X$ such that $s(u)=s'(u)$ but
$s(v)\neq s'(v)$. Since $x$ and $u$ determine $v$
it then follows that $s(x)\neq s'(x)$. Using independence,
we infer that there exists $s''\in X$ with $s''(u)=s(u)=s'(u)$,
$s''(x)=s(x)$, and $s''(v)=s'(v)\neq s(v)$ contradicting
the assumption that $x$ and $u$ determine $v$.

\medskip\noindent{\bf Remark. } In item (2) we might
use instead of $v\ \bot_{u\vz}\ x$ also the stronger condition
$uv\ \bot_{\vz}\ x$. On the other side the weaker condition
$v\ \bot_{\vz}\ x$ does not suffice. At first sight, 
implications of the form  $(\=(x,u,v) \land v\bot x)\ \Rightarrow\ \=(u,v)$
(saying that if a variable is completely determined
by $x$ and $u$ and independent from $x$, then it
is already determined by $u$) might seem plausible.
However, it is not valid in general, and it is quite easy to
construct a counterexample over a domain of three values.

An interesting question raised by the discussion above concerns the
power of independence logic for expressing properties of teams.
Given that independence logic can be embedded into existential second-order logic, and that, on finite structures,
existential second-order logic captures the complexity class NP,
we may formulate this for finite structures as follows:

\medskip
\noindent{\bf Problem:\ }Characterize the NP properties of teams that correspond to formulas of independence logic.

\medskip

\noindent
Note that the corresponding question for dependence logic
is solved by Theorem 13: Dependence logic can define exactly those NP properties of teams that are expressible in existential second-order logic with the predicate for the team occurring only negatively.
Very recently, Pietro Galliani \cite{PG} has also solved the problem for independence logic, showing that actually \emph{all} NP-properties of
teams can be defined in independence logic.

\def\Dbar{\leavevmode\lower.6ex\hbox to 0pt{\hskip-.23ex \accent"16\hss}D}
  \def\cprime{$'$}


\begin{thebibliography}{1}

\bibitem{armstrong}
W.~W. Armstrong.
\newblock Dependency structures of data base relationships.
\newblock {\em Information Processing}, 74, 1974.

\bibitem{CH}
Ashok K. Chandra and Moshe Y. Vardi. 
The implication problem for functional and inclusion dependencies is undecidable. 
SIAM Journal on Computing, 14(3):671Ð677, 1985.

\bibitem{PG}
Pietro Galliani.
Inclusion and Exclusion in Team Semantics: On some logics of imperfect information. 
Annals of Pure and Applied Logic, 163(1):68-84, 2012.

\bibitem{MR97j:03005}
Jaakko Hintikka.
\newblock {\em The principles of mathematics revisited}.
\newblock Cambridge University Press, Cambridge, 1996.
\newblock Appendix by Gabriel Sandu.

\bibitem{MR98k:03068}
Wilfrid Hodges.
\newblock Compositional semantics for a language of imperfect information.
\newblock {\em Log. J. IGPL}, 5(4):539--563 (electronic), 1997.

\bibitem{kont}
Juha Kontinen and Jouko V{\"a}{\"a}n{\"a}nen.
\newblock On definability in dependence logic.
\newblock {\em J. Log. Lang. Inf.}, 18(3):317--332, 2009.

\bibitem{MR2351449}
Jouko V{\"a}{\"a}n{\"a}nen.
\newblock {\em Dependence logic}, volume~70 of {\em London Mathematical Society
  Student Texts}.
\newblock Cambridge University Press, Cambridge, 2007.

\end{thebibliography}

\end{document}